\documentclass[journal]{IEEEtran}
\usepackage{cite}
\usepackage{pifont,amssymb,url}
\usepackage{multicol,mathrsfs}
\usepackage[usenames,dvipsnames,svgnames,table]{xcolor}
\usepackage{bm,bbm,mathrsfs,amssymb,pifont,amssymb,pifont,amsfonts,amsmath,stmaryrd,color,amssymb,graphics,graphicx,epstopdf}
\usepackage{amscd,graphicx,array,dsfont,texdraw,tikz,enumerate,multicol,multirow,verbatim}
\usetikzlibrary{arrows,shapes,shadows,positioning,automata,patterns}
\usetikzlibrary{matrix,arrows,calc,trees,decorations.pathmorphing,decorations.markings, decorations.pathreplacing}
\usepackage[T1]{fontenc}
\usepackage{booktabs}
\usepackage{algorithm,algorithmicx,algpseudocode}
\usepackage{amsmath}

\usetikzlibrary{arrows.meta}
\usetikzlibrary{decorations.pathreplacing}

\newcommand{\dd}{{\rm d}\hbox{\hskip 0.5pt}}

\newcommand{\bb}{\color{black}}
\newcommand{\bbb}{\color{black}}

\newcommand{\bremark}{\smallskip\begin{remark}\begin{rm}}
\newcommand{\eremark}{\end{rm}\hfill$\boxempty$\end{remark}\smallskip}
\newcommand{\btheorem}{\smallskip\begin{theorem} \begin{it}}
\newcommand{\etheorem}{\end{it}\end{theorem}\smallskip}
\newcommand{\blemma}{\smallskip\begin{lemma} \begin{it} }
\newcommand{\elemma}{\end{it}\end{lemma}\smallskip}
\newcommand{\bcorollary}{\smallskip\begin{corollary} \begin{it} }
\newcommand{\ecorollary}{\end{it}\end{corollary}\smallskip}
\newcommand{\bdefinition}{\smallskip\begin{definition}\begin{rm}}
\newcommand{\edefinition}{\end{rm}\hfill$\boxempty$\end{definition}\smallskip}
\newcommand{\bproposition}{\smallskip\begin{proposition}\begin{it}}
\newcommand{\eproposition}{\end{it}\end{proposition}\smallskip }
\newcommand{\bexample}{\smallskip\begin{example}\begin{rm}}
\newcommand{\eexample}{\end{rm}\hfill$\boxempty$\end{example}\smallskip}
\newcommand{\bproblem}{\smallskip\begin{problem}\begin{rm}}
\newcommand{\eproblem}{\end{rm}\hfill$\boxempty$\end{problem}\smallskip}

\newcommand{\bassume}{\smallskip\begin{assumption}\begin{rm}}
\newcommand{\eassume}{ \end{rm}\hfill$\boxempty$\end{assumption}\smallskip}
\newcommand{\bfact}{\smallskip\begin{fact}\begin{it}}
\newcommand{\efact}{ \end{it} \end{fact}\smallskip}
\newcommand{\bcondition}{\begin{condition}\begin{rm}}
\newcommand{\econdition}{ \end{rm}\hfill$\boxempty$\end{condition}\smallskip}

\newtheorem{theorem}{Theorem}[section]
\newtheorem{lemma}{Lemma}[section]
\newtheorem{corollary}{Corollary}[section]
\newtheorem{definition}{Definition}[section]
\newtheorem{proposition}{Proposition}[section]
\newtheorem{problem}{Problem}[section]
\newtheorem{property}{Property}[section]

\newtheorem{myremark}{Remark}[section]
\newenvironment{remark}{\begin{myremark}\normalfont}{\end{myremark}}
\newtheorem{myexample}{Example}[section]
\newenvironment{example}{\begin{myexample}\normalfont}{\end{myexample}}
\newtheorem{assumption}{Assumption}
\newtheorem{fact}{Fact}[section]
\newtheorem{condition}{\it Condition}[section]

\newcommand{\bbm}[1]{\left[\begin{matrix} #1 \end{matrix}\right]}

\newcommand{\bproof}{ \textit{Proof:} \begin{rm} }
\newcommand{\eproof}{ \end{rm} \hfill Q.E.D.}  

\newcommand{\rline}{{\mathbb R}}

\newenvironment{proof}{\vspace{.1cm}\noindent{\sc Proof.}\hspace{0.10cm}\,\,}{$\hfill\Box$\vspace{.1cm}}

\newcommand{\rfb}[1]{\mbox{\rm
   (\ref{#1})}\ifx\undefined\stillediting\else:\fbox{$#1$}\fi}
   
\newcommand{\bluff}{{\hbox{\raise 15pt \hbox{\hskip 0.5pt}}}}

\newfont{\roma}{cmr10 scaled 1200}

\usepackage{hyperref}
\hypersetup{hypertex=true,
 	colorlinks=true,
 	linkcolor=blue,
 	anchorcolor=blue,
       citecolor=blue}

\begin{document}

\title{Distributed formation control of end-effector of mixed planar fully- and under-actuated manipulators}

\author{Zhiyu Peng, Bayu Jayawardhana, Xin Xin

\thanks{Zhiyu Peng is with School of Automation, Southeast University, Nanjing 210096, China (e-mails: 230208667@seu.edu.cn). Bayu Jayawardhana is with Engineering and Technology Institute Groningen, Faculty of Science and Engineering, University of Groningen, Groningen 9747 AG, The Netherlands (e-mails: b.jayawardhana@rug.nl).
Xin Xin is with Faculty of Computer Science and Systems Engineering, Okayama Prefectural University, Okayama 719-1197, Japan (e-mail: xxin@cse.oka-pu.ac.jp). The work of Z. Peng is supported by the China Scholarship Council under Grant No.~202206090189. The work of X. Xin is supported by National Natural Science Foundation of China under Grant No.~61973077.
}
}

\maketitle
\thispagestyle{empty}

\begin{abstract}
This paper addresses the problem of end-effector
formation control for a mixed group of two-link manipulators moving in a horizontal plane that comprises of fully-actuated manipulators and underactuated manipulators with only the second joint being actuated (referred to as the passive-active (PA) manipulators). The problem is solved by extending the distributed end-effector formation controller for the fully-actuated manipulator to the PA manipulator moving in a horizontal plane by using its integrability. This paper presents stability analysis of the closed-loop systems under a given necessary condition, and we prove that the manipulators' end-effector converge to the desired formation shape. The proposed method is validated by simulations. 
\end{abstract}


\section{Introduction}
\label{sec:introduction}

\IEEEPARstart{R}ECENTLY, the distributed formation control for manipulators has attracted significant interests, which allows 
a group of industrial manipulators to collectively carry out a complex task. For example, Wu et al. \cite{wu2022distributed} investigate the distributed end-effector formation control of fully-actuated manipulators (the number of whose inputs is equal to its degree-of-freedom). The results are based on the use of virtual springs between the edges of an infinitesimally rigid formation graph of end-effectors. However, when underactuated manipulators (which have fewer inputs than the degree-of-freedom \cite{oriolo1991control, Spo95,fantoni2000energy,wu2019general,lai2015stabilization}) are used for some of the agents, the results are no longer applicable. In this case, it remains an open problem whether a desired formation shape can be made attractive by distributed control laws.

Underactuated manipulators have a wide range of applications due to the cheap cost and simple structure. Furthermore, one can regard a fully-actuated manipulator as underactuated when some of its joint actuators are faulty. In this regards, the control of an underactuated manipulator has been a central research topic for the past decades, which is particularly challenging owing to the second-order nonholonomic constraints \cite{FL01,XL14}. 

In this paper, we study the distributed end-effector formation control for a group of two-link manipulators moving in a horizontal plane, whose motions are not affected by gravity. Different from the work in \cite{wu2022distributed}, we consider that some manipulators in the group only have a single actuator at the second joint. These manipulators are referred to as the passive-active (PA) manipulator (see Fig. \ref{fig1}), which is a typical planar underactuated manipulator.

In our main results, we tackle some new difficulties in controller design and stability analysis. {\bb Firstly, we extend the distributed end-effector formation control for a group of fully-actuated manipulators in \cite{wu2022distributed} to the above-mentioned mixed group of manipulators.} For the PA manipulator with zero initial joint velocity (or joint angular velocity), the authors in \cite{oriolo1991control} demonstrate its integrability and show that it is holonomic. That means the end-effector of the PA manipulator (with zero initial joint velocity) actually moves in a (curved) line rather than in a plane, and the end-effector position depends entirely on its actuated joint position (or joint angle). Based on the integrability, Lai et al. \cite{lai2015stabilization} derive constraints on the joint position/velocity of the PA manipulator with zero initial joint velocity. For solving the distributed formation control of fully-actuated manipulators' end-effector, Wu et al. \cite{wu2022distributed} use the Jacobian matrix of the fully-actuated manipulator which is obtained by taking the partial derivative of its end-effector position with respect to its joint position. Inspired by the result in \cite{wu2022distributed}, we are able to obtain the Jacobian matrix of the PA manipulator (with zero initial joint velocity) without relying on the computation of the derivative due to the results in \cite{oriolo1991control} and \cite{lai2015stabilization}.

Secondly, we present a mild condition such that the proposed distributed end-effector formation control laws are admissible, and subsequently we conduct the stability analysis for the closed-loop systems. 
In \cite{wu2022distributed}, Wu et al. assume that the Jacobian matrix of every fully-actuated manipulator is full rank, {\bb however, this condition can not be fulfilled by the PA manipulator since its Jacobian matrix reduces to a column vector.} For the PA manipulator with zero initial joint velocity, 
we find another equation to obtain a square augmented Jacobian matrix. Under the condition of all obtained augmented Jacobian matrices being invertible and all fully-actuated manipulators' Jacobian matrix being full rank, we prove that all manipulators' end-effector converge to the desired formation shape. Through a number of simulation results, we show that the presented condition is always satisfied and the proposed control laws are effective if the group of manipulators starts from a neighborhood of the desired and reachable formation shape.

The rest of this paper is organized as follows. We present the manipulator dynamics and kinematics in Section \ref{sec:preliminaries}. Section \ref{sec:preliminaries} also contains some preliminaries on the graph theory, standard distributed formation control, and main problem formulation. Section \ref{main result} presents the distributed end-effector formation control laws and the corresponding stability analysis. Sections \ref{simu} and \ref{conclusion}  give simulation and the conclusions respectively.

\section{Preliminaries and Problem Formulation}
\label{sec:preliminaries}

This paper concerns the distributed end-effector formation control for a group of $N>1$ two-link manipulators moving in the same horizontal $X-Y$ plane.  Consider that $N_1\geq 1$ 
of these manipulators are fully-actuated, and the rest $N_2=N-N_1$ 
manipulators are underactuated with a single actuator at their second joint, which we refer to as the PA manipulators. The manipulators' mechanical parameters are described in Table \ref{para}.

\emph{Notation:} For column vectors $x_1$, ..., $x_n$, let $\operatorname{col}\left(x_1, \ldots, x_n\right): =\left[x_1^\mathrm{T}, \ldots, x_n^\mathrm{T}\right]^\mathrm{T}$ be the
stacked column vector. Use a short-hand notation $\bar {B}:= B\otimes I_2$, where $I_2 \in \mathbb{R}^{2\times2}$ is the identity matrix and $\otimes$ denotes the Kronecker product. 

\subsection{Manipulator Dynamics and Kinematics}
\begin{table}[tp]%
\caption{The mechanical parameters of manipulator $i$.}
\label{para}\centering %
\begin {tabular}{clccc}
\toprule %
		{\bb Symbol ($j=1,2$)} & Description \\
		\midrule
		$m_{i,j}$  & Mass of link $j$\\
		\midrule
		$I_{i,j}$  & Moment of inertia of link $j$ with\\ & respect to its center-of-mass (COM) \\
		\midrule
		$L_{i,j}$  & Length of link $j$ \\
		\midrule
		$l_{i,j}$  & Distance between joint $j$\\ & and the COM of link $j$\\
		\midrule
		$q_{i,j}$  & Angle of joint $j$ \\
		\midrule
		$u_{i,j}$  & Torque applied to joint $j$\\
		\bottomrule
\end {tabular}
\end {table}

\begin{figure}
	\centering
	\includegraphics[scale=0.4]{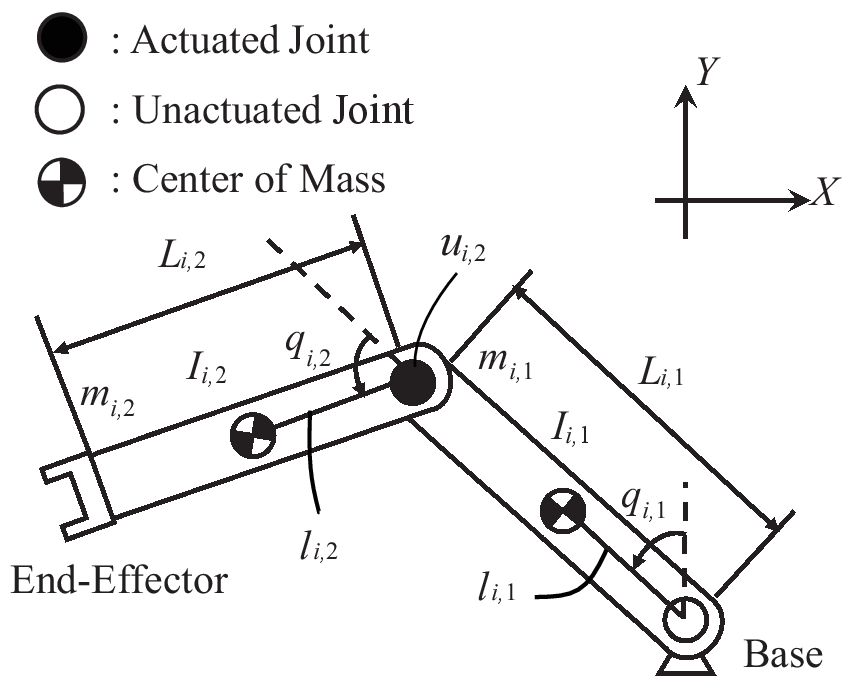}
	\caption{{\bb The PA manipulator: a  planar two-link underactuated manipulator with an unactuated (or passive) first joint and an actuated (or active) second joint.}} 
 \label{fig1}
\end{figure}

Based on the Euler-Lagrange equation \cite{murray2017mathematical, spong2006robot}, the group of two-link planar manipulators is modeled by
\begin{equation}
	M_i\left(q_i\right) \ddot{q}_i+C_i\left(q_i, \dot{q}_i\right) \dot{q}_i= u_i, \label{full}
\end{equation}
where $i \in\{1,...,N\}$, and $q_i=[q_{i,1}, q_{i,2}]^\mathrm{T}$ is the generalized joint position. {\bb Without loss of generality, suppose that for all $i \in \{1,...,N_1\}$, the $i$-th manipulator is fully-actuated and $u_i= [u_{i,1}, u_{i,2}]^\mathrm{T}$ is its generalized forces; otherwise for all $i \in \{N_1+1,...,N\}$, the $i$-th manipulator is an underactuated PA manipulator with $u_i= [0, u_{i,2}]^\mathrm{T}$. }
The matrix $M_i(q_i)$ is the mass matrix and $C_i({q_i},{{\dot q}_i})$ is the Coriolis and centrifugal term, which are respectively given by 
\[\begin{aligned}
M_i\left(q_i\right) & =\left[\begin{array}{lc}
M_{i,11}\left(q_i\right) & M_{i,12}\left(q_i\right) \\
M_{i,21}\left(q_i\right) & M_{i,22}\left(q_i\right)
\end{array}\right] \\
& =\left[\begin{array}{cc}
\alpha_{i,1}+\alpha_{i,2}+2 \alpha_{i,3} \cos q_{i,2} & \alpha_{i,2}+\alpha_{i,3} \cos q_{i,2} \\
\alpha_{i,2}+\alpha_{i,3} \cos q_{i,2} & \alpha_{i,2}
\end{array}\right],
\end{aligned}\]
\[C_i({q_i},{{\dot q}_i}) = {\alpha _{i,3}}\left[ {\begin{array}{*{20}{c}}
		{ - {{\dot q}_{i,2}}\;\;}&{ - {{\dot q}_{i,1}} - {{\dot q}_{i,2}}}\\
		{{{\dot q}_{i,1}}}&0
\end{array}} \right]\sin {q_{i,2}},\]
where the mechanical parameters $\alpha_{i,1}, \alpha_{i,2}, \alpha_{i,3}$ are as follows
\begin{equation}
\left\{\begin{array}{l}
	\alpha_{i,1}=m_{i,1} l_{i,1}^2+m_{i,2} L_{i,1}^2+I_{i,1}, \vspace{0.1cm}\\
	\alpha_{i,2}=m_{i,2} l_{i,2}^2+I_{i,2}, \vspace{0.1cm}\\
	\alpha_{i,3}=m_{i,2} L_{i,1} l_{i,2} .  \label{parameters}
\end{array}\right.
\end{equation}


\begin{property} \label{property}
{\bb Following standard properties of Euler-Lagrange systems \cite{spong2006robot, Ortega-book}, the manipulators 
(\ref{full}) satisfy the following two properties: 
\begin{enumerate}
\item[{\bf P1.}] The matrix $M_i(q_i)$ is positive definite; 	
\item[{\bf P2.}] The matrix $\dot M_i(q_i)-2C_i(q_i, \dot q_i)$ is skew-symmetric. 
\end{enumerate}}
\end{property}

Accordingly, due to {\bf P1}, we can rewrite (\ref{full}) compactly as 
\begin{equation}
	{\ddot q = {M^{ - 1}}(q)\Big(\bar u - C(q,\dot q)\dot q\Big)}, \label{dd}
\end{equation}
where $q=\operatorname{col}\left(q_1, \ldots, q_N\right)\in \mathbb{R}^{2N}$, $\bar u = \operatorname{col}\left(u_1, \ldots, u_N\right)\in \mathbb{R}^{2N}$ are the stacked vectors of $ q_i$ and $ u_i$ respectively and the matrices $ M(q), C(q,\dot q) \in \mathbb{R}^{2N \times 2N}$ are the block diagonal matrices of all $M_i(q_i), C_i(q_i,\dot q_i)$, respectively.

For manipulator $i$, let $x_i(t) \in \mathbb{R}^2$ be its end-effector position in the task-space, which can be obtained by
\begin{equation}
	x_i: = [x_{i,X}, x_{i,Y} ]^\mathrm{T} = h_i\left(q_i\right)+x_{i0}, \qquad  \label{task-space} 
\end{equation}
where $x_{i0} \in \mathbb{R}^2$ is the position of the fixed manipulator base and 
\begin{equation*}
	h_i\left(q_i\right) = \left[ {\begin{array}{*{20}{c}}
		{-{L_{i,1}}\sin \left( {{q_{i,1}}} \right) - {L_{i,2}}\sin \left( {{q_{i,1}} + {q_{i,2}}} \right)}\\
		{{L_{i,1}}\cos \left( {{q_{i,1}}} \right) + {L_{i,2}}\cos \left( {{q_{i,1}} + {q_{i,2}}} \right)}
\end{array}} \right]. \label{hk}
\end{equation*}
Differentiating (\ref{task-space}) with respect to time $t$ leads
\begin{equation}
\begin{aligned}
& \dot{x}_i=J_i\left(q_i\right) \dot{q}_i,\\
& J_i\left(q_i\right):=\frac{\partial h_i\left(q_i\right)}{\partial q_{i}}=\left[\begin{array}{ll}
J_{i,11}\left(q_i\right) & J_{i,12}\left(q_i\right) \\
J_{i,21}\left(q_i\right) & J_{i,22}\left(q_i\right)
\end{array}\right],
\end{aligned}\label{velocity}
\end{equation}
where $J_i(q_i)$ is the Jacobian matrix \cite{murray2017mathematical,spong2006robot} of the manipulator $i$, and 
 \begin{align*}
& J_{i,11}\left(q_i\right)= - L_{i,1} \cos \left(q_{i,1}\right) - L_{i,2} \cos \left(q_{i,1}+q_{i,2}\right), \\
& J_{i,12}\left(q_i\right)= - L_{i,2} \cos \left(q_{i,1}+q_{i,2}\right), \\
& J_{i,21}\left(q_i\right)=-L_{i,1} \sin \left(q_{i,1}\right) - L_{i,2} \sin \left(q_{i,1}+q_{i,2}\right), \\
& J_{i,22}\left(q_i\right)= - L_{i,2} \sin \left(q_{i,1}+q_{i,2}\right).
 \end{align*}


\subsection{Properties of the Underactuated PA Manipulator}

Let us recall some properties of  the  PA manipulator $i \in \{N_1+1,...,N\}$ given by (\ref{full}) with $u_i= [0, u_{i,2}]^\mathrm{T}$. We denote  $q_{i}(0)=[q_{i,1}(0), q_{i,2}(0)]^\mathrm{T}$ and $\dot q_{i}(0)=[\dot q_{i,1}(0), \dot q_{i,2}(0)]^\mathrm{T}$ as its initial joint position and initial joint velocity, respectively. Firstly, recall an important property on its joint position $q_{i}$ and joint velocity $\dot q_{i}$, which is presented by the following lemma and we refer to the results in  \cite{oriolo1991control,lai2015stabilization} for detailed discussion.\vspace{0.2cm}

\begin{lemma}
	\cite{oriolo1991control,lai2015stabilization}
 If the  PA manipulator $i$ starts from a stationary position, i.e. $\dot q_{i}(0)=0$, then its joint velocity $ \dot q_{i}$ satisfies
	\begin{equation}
M_{i,11}(q_i) \dot{q}_{i,1}+M_{i,12}(q_i) \dot{q}_{i,2}=0.   \label{speed}
	\end{equation}
Furthermore, assume that $q_{i,2}(0) \in [-\pi, \pi]$ and $q_{i,2} \in [-\pi+2k\pi, \pi+2k\pi]$, $k \in \mathbb{Z}$. Then the system is holonomic and its unactuated joint position $q_{i,1}$ depends entirely on its actuated joint position $q_{i,2}$ by 
	\begin{equation}
{q_{i,1}} = f(q_{i,2}) =  - \frac{{{q_{i,2}}}}{2} - \gamma \arctan \left( {\rho \tan \frac{{{q_{i,2}}}}{2}} \right) -\gamma k \pi + \eta,\label{angle constraint}
	\end{equation}
where
 \begin{align*}
 \gamma  = \frac{{{\alpha_{i,2}} - {\alpha_{i,1}}}}{{\sqrt {{{\left( {{\alpha_{i,1}} + {\alpha_{i,2}}} \right)}^2} - 4\alpha_{i,3}^2} }}, \; \rho  = \sqrt {\frac{{{\alpha_{i,1}} + {\alpha_{i,2}} - 2{\alpha_{i,3}}}}{{{\alpha_{i,1}} + {\alpha_{i,2}} + 2{\alpha_{i,3}}}}}, 
  \end{align*}
and $\displaystyle \eta = \frac{{{q_{i,2}}(0)}}{2} + {q_{i,1}}(0) + \gamma \arctan \left( \rho \tan \frac{{{q_{i,2}}(0)}}{2} \right)$. 
	\label{constraint}
\end{lemma}
\hspace*{\fill}\vspace{0.2cm}

Now, for the PA manipulator $i$, we rewrite (\ref{task-space}) and (\ref{velocity}) according to Lemma \ref{constraint}. Substituting (\ref{angle constraint}) into (\ref{task-space}) gives us
\begin{equation}
   {x_i}: = {\left[ {{x_{i,X}},{x_{i,Y}}} \right]^{\rm{T}}} = \bar h_i\left( {{q_{i,2}}} \right) + {x_{i0}}, \label{task}
\end{equation}
where 
\begin{equation*}
	\bar h_i\left(q_i\right) = \left[ {\begin{array}{*{20}{c}}
		{-{L_{i,1}}\sin \left( f(q_{i,2}) \right) - {L_{i,2}}\sin \left( {f(q_{i,2}) + {q_{i,2}}} \right)}\\
		{{L_{i,1}}\cos \left( f(q_{i,2}) \right) + {L_{i,2}}\cos \left( {f(q_{i,2}) + {q_{i,2}}} \right)}
\end{array}} \right]. \label{hk2}
\end{equation*}
Note that $M_{i,11}(q_i)>0$ due to $M_i(q_i)$ being positive definite as in Property \ref{property}. By using (\ref{speed}), we can rewrite (\ref{velocity}) as  
\begin{equation}
\begin{aligned}
\dot{x}_i & =\left[\begin{array}{ll}
J_{i, 11}\left(q_i\right) & J_{i, 12}\left(q_i\right) \\
J_{i, 21}\left(q_i\right) & J_{i, 22}\left(q_i\right)
\end{array}\right]\left[\begin{array}{l}
\displaystyle\frac{-M_{i, 12}\left(q_i\right)}{M_{i, 11}\left(q_i\right)} \dot{q}_{i, 2} \\
\dot{q}_{i, 2}
\end{array}\right] \\
& =\bar{J}_i\left(q_i\right) \dot{q}_{i, 2},
\end{aligned}
\label{velocity 2}
\end{equation}
where the Jacobian matrix (vector) $\bar{J}_i\left(q_i\right)$ is
 \begin{equation*}
 \begin{aligned}
\bar{J}_i\left(q_i\right) & =\left[\bar{J}_{i, 1}\left(q_i\right), \bar{J}_{i, 2}\left(q_i\right)\right]^{\mathrm{T}} \\
& =\left[\begin{array}{c}
\displaystyle -\frac{J_{i, 11}\left(q_i\right) M_{i, 12}\left(q_i\right)}{M_{i, 11}\left(q_i\right)}+J_{i, 12}\left(q_i\right) \vspace{0.1cm}\\
\displaystyle -\frac{J_{i, 21}\left(q_i\right) M_{i, 12}\left(q_i\right)}{M_{i, 11}\left(q_i\right)}+J_{i, 22}\left(q_i\right)
\end{array}\right].
\end{aligned}
\end{equation*}
Note that we use the Jacobian matrix ${\bar J}_i(q_i)$ in the distributed controller for the  PA manipulator $i$. Two remarks on  $\bar{J}_i(q_i)$ are as follows. 

\begin{remark}
The Jacobian matrix ${\bar J}_i$ can be expressed as a function of 
$q_{i,2}$ only. When the real-time information of unactuated joint position $q_{i,1}$ is not available, we can instead use $q_{i,2}$ to substitute $q_{i,1}$ as given in \eqref{angle constraint}. However, in this particular case, the distributed controller for the  PA manipulator $i$ becomes complicated and is highly dependent on the initial joint position. 
\label{remark}
\end{remark}\vspace{0.2cm}

\begin{remark}
We can also obtain the Jacobian matrix	${\bar J}_i$ by differentiating the end-effector position of the PA manipulator with respect to its actuated
joint position similar to that for the fully-actuated manipulator, i.e. $ {{\bar J}_i} = \frac{{\partial {\bar h_i}\left( {{q_{i,2}}} \right)}}{{\partial {q_{i,2}}}}$. However, this approach is practically not feasible as it introduces computational complexity and makes the controller particularly complicated.
\end{remark}\vspace{0.2cm}

Define $\mathcal{W}_i$ as the working space of manipulator $i \in \{1,...,N\}$, and the 
entire working space for the networked manipulators is $\mathcal{W}:=\mathcal{W}_1 \times \cdots \times \mathcal{W}_N$. According to (\ref{task-space}), for the fully-actuated manipulator  $i \in \{1,...,N_1\}$ modeled by (\ref{full}) with $u_i= [u_{i,1}, u_{i,2}]^\mathrm{T}$, we have
\begin{equation}
 \mathcal{W}_i \subset\left\{x_i \in \mathbb{R}^2: x_i=h_i\left(q_i\right)+x_{i0}, \; q_i \in \mathbb{R}^2\right\}.
\end{equation}
Consider the  PA manipulator  $i \in \{N_1+1,...,N\}$ modeled by (\ref{full}) with $u_i= [0, u_{i,2}]^\mathrm{T}$, which starts from $\dot q_i(0)=0$. Assume that $q_{i,2}(0) \in [-\pi, \pi]$ and $q_{i,2} \in [-\pi+2k\pi, \pi+2k\pi]$, $k\in \mathbb{Z}$. According to (\ref{task}), we have
\begin{equation}
\mathcal{W}_i \subset\left\{x_i \in \mathbb{R}^2: x_i= \bar h_i\left(q_{i,2}\right)+x_{i0},\; q_{i,2} \in \mathbb{R} \right\}. \label{Wj}
\end{equation}
Note that, unlike the fully-actuated manipulator, the working space of the PA manipulator is in a line rather than a plane.

\subsection{Formation Graph and Mixed End-effector Distributed Formation Control Problem}
For a given desired geometrical formation shape of manipulators' end-effector, we can associate an undirected graph to the vertices and edges of the formation shape. Let us describe the corresponding 
formation graph by  $\mathcal{G}: =\{\mathcal{V}, \mathcal{E}\}$, where $\mathcal{V}:=\{1, \cdots, N\}$ is the vertex set and $\mathcal{E} \subset \mathcal{V} \times \mathcal{V}$ is the ordered edge set with $\mathcal{E}_k$ denoting the $k$-th edge. The numbers of vertices and edges of $\mathcal{G}$ are $|\mathcal{V}| = N$ and $|\mathcal{E}|$, respectively. 
{\bb The set of edges, where the end-effector $i$ is part of, is given by $\mathcal{I}_i:=\{k\in \{1,\ldots,|\mathcal E|\}:(i, j) = \mathcal{E}_k \text{ for some }j\}$.} 
We define the elements of the incidence matrix $B \in \mathbb{R}^{N \times|\mathcal{E}|}$ of $\mathcal{G}$ by
\begin{equation}
b_{i k}=\left\{\begin{array}{ll}
+1, & i=\mathcal{E}_k^{\text {tail }} \vspace{0.2cm}\\
-1, & i=\mathcal{E}_k^{\text {head }} \vspace{0.2cm}\\
0, & \text {otherwise}
\end{array} \; i = 1, \ldots,N, \; k=1, \ldots,|\mathcal{E}|,\right.
\end{equation}
where $\mathcal{E}_k^{\text {tail }}$ and $\mathcal{E}_k^{\text {head}}$ denote the tail and head nodes, respectively, of $\mathcal{E}_k$, i.e. $(\mathcal{E}_k^{\text {tail }}, \mathcal{E}_k^{\text {head}})=\mathcal{E}_k$. 
Let us stack all end-effectors’ position $x_i$ into $x=\operatorname{col}\left(x_1, \ldots, x_N\right) \in \mathbb{R}^{2N}$ and define the relative displacement $z$ by $z=\bar{B}^{\rm T}x$, whose elements correspond to the ordered $z_k:=x_i-x_j$. {\bb For a given desired formation shape defined by a vector of desired distances $d^*$ on the edges, let $x^*$ define a reference position such that $\|z_k(x^*)\| = d_k^*$ for all $k \in \{1,...,|\mathcal E|\}$, where $d_k^*$ denotes the $k$-th element of $d^*$, i.e. the desired distance of the edge $\mathcal{E}_k$.} 
Define the edge function $f_{\mathcal G}(x):=\mathop{\text{col}}\limits_{k\in \{1,\ldots,|\mathcal E|\}}\left(\|z_k\|^2\right)$, and the framework $(\mathcal G,x^*)$ is called {\it infinitesimally rigid} if the rank of $\frac{\dd f_{\mathcal G}}{\dd x}(x^*)$ is $2N-3$ (for 2D shape). We refer interested readers to \cite{Chan2021,Marina2016,Marina2018} and references therein for exposition on the rigid formation graph. Using the reference position $x^*$, let the set of all desired shapes be
\begin{equation}
\mathcal{S}:=\left\{x: x=\left(I_N \otimes R\right) x^*+\mathbf{1}_N \otimes b, R \in \mathbf{S O}(2), b \in \mathbb{R}^2\right\},
\end{equation}
where $\mathbf{1}_N \in \mathbb{R}^N$ denotes the vector whose all elements are one. 
Let $\mathcal{S}_W:=\mathcal{S} \cap \mathcal{W}$ be the set of desired shapes that are also reachable by the networked manipulators. \vspace{0.2cm}

\begin{problem}
{\bf (Mixed Planar {\bb Fully-actuated and PA Manipulators' End-Effector}  Distributed Formation Control Problem)} 
For the above setup of networked two-link manipulators with mixed  fully-actuated and PA manipulators and for a given desired infinitesimally rigid formation shape defined by the framework $(\mathcal G,x^*)$, design the distributed controller of the form
\begin{equation}
u_i = \sigma(\{z_k\}_{k\in \mathcal I_i},q_i,\dot q_i), \;\; i=1,...,N,
\end{equation}
such that $x(t) \to \mathcal{S}_W$ and $\dot q(t) \to 0$ as $t \to \infty$. 
\label{problem}
\end{problem}

Note that due to the nature of the  PA manipulators, the control inputs that are used for them are only those for the second joint. Hence for all $i\in \{N_1+1,\ldots, N\}$, the first element of $u_i$ is not used at all. 

\section{Proposed Distributed Formation Controller} \label{main result}

In the following, we will follow and modify accordingly the end-effector distributed formation controller presented in \cite{wu2022distributed}. For a given desired distance vector $d^*$ associated to an infinitesimally rigid formation shape, we define the error $e_k \in \mathbb{R}$ on the edge $\mathcal{E}_k$ by
$e_k: = \left\|z_k\right\|^2-(d_k^*)^2$. 
Based on the error vector $e=\operatorname{col}(e_1,\ldots,e_{|\mathcal E|})$, we define the potential function $V (e)$ for the formation by
\begin{equation}
V(e) := 
\frac{1}{2}\sum\limits_{k  = 1}^{\left| {\cal E} \right|} {e_k^2}. \label{Ve}
\end{equation}

In order to define the gradient-based distributed control for every manipulator $i \in \{1,...,N\}$, let $\hat e_i \in \mathbb{R}^{2} $ be the gradient of $V(e)$ along $x_i$, i.e. $\hat e_i: = \frac{\partial V(e)}{\partial x_{i}}$, which is expressed by the local relative displacement $z_k$ and the distance error $e_k$ for all $k\in \mathcal I_i$. 
Note that $\frac{{\partial {z_k }}}{{\partial {x_i }}} = {b_{i k}}$ and let ${D_k }\left( {{z_k }} \right) := \frac{{\partial {e_k }}}{{\partial {z_k }}} = 2{z_k }$. Routine computation shows that 
\begin{equation}
{\hat e_i} = \sum\limits_{k = 1}^{|{\cal E}|} {{b_{i k }}} {D_k }\left( {{z_k }} \right){e_k }. 
\label{hat e1}
\end{equation}
We can rewrite (\ref{hat e1}) in the following compact form 
\begin{equation}
\hat{e}=\frac{\partial V(e)}{\partial x} = \bar BD(z)e,  \label{defined}
\end{equation}
where $\hat e=\operatorname{col}\left(\hat e_1, \ldots, \hat e_{N}\right) \in \mathbb{R}^{2N}$ and the matrix $D(z)\in \mathbb{R}^{2|{\cal E}|\times |{\cal E}|}$ is the block diagonal matrix of $D_k(z_k)$ for all $k\in \{1,\ldots,|\mathcal E|\}$. As discussed in \cite{Marina2016}, the matrix $D^\mathrm{T}(z)\bar B^\mathrm{T} \bar BD(z)$ is positive definite if $\mathcal{G}$ is infinitesimally and minimally rigid. 

We are now ready to study the solvability of Problem \ref{problem} and to present the distributed formation controller. 
Prior to that, we need the following assumptions.\vspace{0.2cm}

\begin{assumption}
All PA manipulators start from a stationary position, i.e. $\dot q_i(0)=0$ for all $i\in\{N_1+1,\ldots, N\}$.  \label{initial}
\end{assumption}\vspace{0.2cm}

Let us briefly remark on this assumption. For the  PA manipulator $i$, (\ref{speed}) is only satisfied under Assumption \ref{initial}. When $\dot q_i(0) \ne 0$, we can not guarantee that the  PA manipulator can be stabilized at a target position \cite{lai2015stabilization}. Assumption \ref{initial} is not restrictive since we can always make a manipulator start from a stationary position. 

\begin{assumption}
There exists {\bb a neighborhood $\mathcal S_{W_r}$ of $\mathcal S_W$} 
s.t. 
\begin{description}
\item[{\bf A1.}] for every fully-actuated manipulator $i\in\{1,\ldots,N_1\}$, 
the Jacobian matrix $J_i(q_i)$ as in (\ref{velocity}) is full rank; and 
\item[{\bf A2.}] for every  PA manipulator $i\in\{N_1+1,\ldots,N\}$, 
the Jacobian matrix $\bar J_i(q_i)$ as in (\ref{velocity 2}) satisfies 
${\bar J_{i,1}}(q_i) {\bar J_{i,2}}(q_i) \ne 0$
\end{description}
for all $x \in \mathcal S_{W_r}$. \label{Ass:3}
\end{assumption}\vspace{0.2cm}

Firstly, we note that for {\bf A1}, the Jacobian  matrix $J_i(q_i)$ is not full rank only at $q_{i,2}=k\pi$ ($k \in \mathbb{Z}$) \cite[p.~21]{spong2006robot}. 
Secondly, for the fulfillment of {\bf A2}, the equation $\bar J_{i,1}(q_i) \bar J_{i,2}(q_i) = 0$ only holds at some isolated points, which can be calculated numerically as we show later in the simulation. Assumption \ref{Ass:3} 
is satisfied by continuity argument when the
desired shape has been chosen at a reference point $x^*$ where $q_i=q_i^*$ for all $i \in \{1,...,N\}$, such that $J_i(q_i^*)$ is full rank (for all $i \in \{1,...,N_1\}$) and $\bar J_{i,1}(q_i^*) \bar J_{i,2}(q_i^*) \ne 0$ (for all $i \in \{N_1+1,...,N\}$) in $\mathcal S_W$. \vspace{0.2cm}  

\begin{theorem}
	\label{the}
Consider {\bb the end-effector distributed formation control problem of mixed planar fully-actuated and PA manipulators} in Problem \ref{problem}. Under Assumptions \ref{initial} and \ref{Ass:3}, the problem can be solved locally 
by distributed formation control laws 
\begin{equation}
u_i = \left\{\begin{array}{rl}
- {K_P}J_i^\mathrm{T}(q_i){\hat e_i} - {K_D} \dot q_{i}, &  i = 1,...,N_1, \\ 
& \\
\bbm{0 \\ -K_P \bar{J}_i^{\mathrm{T}}(q_i) \hat{e}_i-K_D \dot{q}_{i,2}}, 
& i=N_1+1,\ldots,N,
\end{array} \right. \label{eq:controller}
\end{equation}	
where $K_P > 0$, $K_D > 0$ are controller gains, the Jacobian matrices $J_i(q_i)$, $\bar J_i(q_i)$ and the vector ${\hat e_i}$ are as in (\ref{velocity}), (\ref{velocity 2}) and (\ref{hat e1}), respectively.
\end{theorem}

\begin{proof}
Firstly, let us rewrite controllers \eqref{eq:controller}
into the following compact form
	\begin{equation}
\begin{aligned}
& u=-K_P J^\mathrm{T} (q) \hat{e}-K_D \xi,  \\
& J^ \mathrm{T}(q) =\left[\begin{array}{cc}
J_{\text {fa }}^\mathrm{T} \left(q_{\text {fa }}\right) & 0 \\
0 & \bar{J}_{\text {pa }}^\mathrm{T} \left(q_{\text {pa}}\right)
\end{array}\right],
\end{aligned}\label{controller}
\end{equation}
{where ``fa'' and ``pa'' refer to the fully-actuated manipulators and the PA manipulators, respectively.} The matrix $J_\text{fa}(q_\text{fa}) \in \mathbb{R}^{{2N_1} \times {2N_1}}$ is the block diagonal matrix of $J_i(q_i)$ for all $i \in \{1,...,N_1\}$ and $\bar J_\text{pa}(q_{\text{pa}}) \in \mathbb{R}^{{2N_2} \times {N_2}} $ ($N_2=N-N_1$) is the block diagonal matrix of $\bar J_i(q_{i})$ for all $i \in \{N_1+1,...,N\}$. The stacked vectors $u,\xi,q_{\text {fa}},q_{\text {pa}}$ are respectively given by
\[
\begin{aligned}
& u=\operatorname{col}\left(u_1, \ldots, u_{N_1}, u_{N_1+1,2}, \ldots, u_{N,2}\right) \in \mathbb{R}^{2N_1+N_2}, \\
& \xi=\operatorname{col}\left(\dot{q}_1, \ldots, \dot{q}_{N_1}, \dot{q}_{N_1+1,2}, \ldots, \dot{q}_{N,2}\right) \in \mathbb{R}^{2N_1+N_2},\\
& q_{\text {fa}}=\operatorname{col}\left(q_1, \ldots, q_{N_1}\right) \in \mathbb{R}^{2 N_1}, \\
& q_{\text {pa}}=\operatorname{col}\left(q_{N_1+1}, \ldots, q_{N}\right) \in \mathbb{R}^{2 N_2}.
\end{aligned}
\]

Consider the  Lyapunov function as follows
	\begin{equation}
		U ={K_P}V(e) + \frac{1}{2}{\dot q ^\mathrm{T} }M(q){\dot q}, \label{U}
	\end{equation}
where $V(e)$ is as in (\ref{Ve}). Routine computation to the time derivative of (\ref{U}) yields
\begin{align}
\nonumber \dot U & = {K_P}{\left( {\frac{{\partial V}}{{\partial x}}} \right)^\mathrm{T}}\dot x + {{\dot q}^ \mathrm{T}}(M(q)\ddot q) + \frac{1}{2}{{\dot q}^\mathrm{T} }\dot M(q)\dot q \\
\nonumber & =  {K_P} {{\hat e}^\mathrm{T} }J(q)\xi  + {{\dot q}^\mathrm{T} }(\bar u - C(q,\dot q)\dot q) + \frac{1}{2}{{\dot q}^\mathrm{T} }\dot M(q)\dot q \\
& = {K_P} {{\hat e}^\mathrm{T} }J(q)\xi  + {{\dot q}^\mathrm{T} }\bar u ={K_P} {{\hat e}^\mathrm{T} }J(q)\xi  + {u^\mathrm{T} }\xi, 
\label{dot U}
\end{align}
where the second equality is due to  
(\ref{defined}), (\ref{velocity}), (\ref{velocity 2}) and (\ref{dd}), and the third equality is due to {\bf P2} of Property \ref{property}. Substituting (\ref{controller}) into {\bb (\ref{dot U})} yields
\begin{equation}
	\dot U = - {K_D}\|\xi\|^2. \label{dot U2}
\end{equation}
It follows from the properness of $U$ and {\bb (\ref{dot U2})} that $e$, $\dot q$ are bounded and $\xi\in L^2(\rline_+)$. Correspondingly, the boundedness of $e$ and $\dot q$ also implies that $\dot\xi \in L^\infty(\rline_+)$, i.e. $\xi$ is uniformly continuous. By the generalized Barbalat's lemma \cite[Theorem 4.4]{Logemann2004},  it implies that $\xi(t)\to 0$ as $t\to \infty$. Under Assumption \ref{initial} and using Lemma \ref{constraint}, we have $\dot q(t)\to 0$ as $t\to\infty$, and consequently, from \eqref{dd}, $\bar u(t)\to 0$ as $t\to \infty$.

Accordingly, the asymptote of \eqref{controller} satisfies  
\begin{equation}
 - {K_P}{J^\mathrm{T}(q) }\hat e = 0. \label{eq mother}
\end{equation}
Let us now decompose (\ref{eq mother}) into
\begin{equation}
- {K_P}J_i^\mathrm{T}(q_i) {{\hat e}_i} = 0, \quad i = 1,...,N_1, \label{eq1}
\end{equation}
and 
\begin{equation}
- {K_P} \bar J_i^\mathrm{T}(q_i) {{\hat e}_i} = 0, \quad i = N_1+1,...,N.\label{eq2}
\end{equation}
On the one hand, under {\bf A1} in Assumption \ref{Ass:3}, 
$J_i(q_i)$ is full rank in $S_{W_r}$, so that $\hat e_i =0$ for all $i \in \{1,...,N_1\}$ from (\ref{eq1}). 
On the other hand, since $\bar J_i(q_i)$ is not full rank, we need an additional equation to complete it. 
Note that for the PA manipulator $i$ with $\dot q_i(0)=0$, the end-effector position $x_i$ depends entirely on its actuated joint position $q_{i,2}$. Thus we have
\begin{equation}
\frac{{\partial V(e)}}{{\partial {q_{i,2}}}} = \frac{{\partial V(e)}}{{\partial {x_{i,X}}}}\frac{{\partial {x_{i,X}}}}{{\partial {q_{i,2}}}} = \frac{{\partial V(e)}}{{\partial {x_{i,Y}}}}\frac{{\partial {x_{i,Y}}}}{{\partial {q_{i,2}}}}. \label{deri}
\end{equation}
Notice that ${{\hat e}_{i,1}} = \frac{{\partial V(e)}}{{\partial {x_{i,X}}}}$, ${{\hat e}_{i,2}} = \frac{{\partial V(e)}}{{\partial {x_{i,Y}}}}$, $\bar J_{i,1}(q_i) = \frac{{\partial {x_{i,X}}}}{{\partial {q_{i,2}}}}$, $\bar J_{i,2}(q_i) = \frac{{\partial {x_{i,Y}}}}{{\partial {q_{i,2}}}}$, so that we can rewrite (\ref{deri}) into 
\begin{equation}
{{\bar J}_{i,1}(q_i)}{{\hat e}_{i,1}} - {{\bar J}_{i,2}(q_i)}{{\hat e}_{i,2}} = 0. \label{supply}
\end{equation}
Combining (\ref{eq2}) and (\ref{supply}) yields 
\begin{equation}
\bar{J}_i^*(q_i) \hat{e}_i=0, \; \bar{J}_i^*(q_i):=\left[\begin{array}{cc}
\bar{J}_{i,1}(q_i) & \bar{J}_{i,2}(q_i) \\
\bar{J}_{i,1}(q_i) & -\bar{J}_{i,2}(q_i)
\end{array}\right].
\end{equation}
The matrix ${\bar J_i ^ *(q_i)}$ is invertible if and only if ${\bar J_{i,1}}(q_i) {\bar J_{i,2}}(q_i) \ne 0$. Then, under {\bf A2} in Assumption \ref{Ass:3}, we have $\hat e_i=0$  for all $i\in \{N_1+1,...,N\}$ from (\ref{eq2}). 
Therefore, 
we have $\hat e =0$. Since $ D^\mathrm{T}(z) \bar B^\mathrm{T}$ is full rank, it follows immediately from \eqref{defined} that in the asymptote we have $e=0$, i.e. $e(t)\to 0$ as $t\to\infty$. 
\end{proof}

\section{Simulation Results}
\label{simu}
We validate the distributed formation controller \eqref{eq:controller} by several numerical simulations in this section. We consider a network of $N = 4$ two-link manipulators moving in the horizontal $X-Y$ plane. The mechanical parameters of these manipulators are the same as those in \cite{wu2022distributed}, where $m_{i,1}=1.2$ kg, $m_{i,2}=1.0$ kg, $l_{i,1}= l_{i,2}=0.75$ m, $L_{i,1}=L_{i,2}=1.5$ m, $I_{i,1} = 0.2250\; \mathrm{kg}  \mathrm{m^2}$ and $I_{i,2}=0.1875$ ${\mathrm{kg}}\mathrm{m^2}$ for $i= 1,2,3,4$.

Let the desired formation shape be a square with side length of $0.4$ m. The incidence matrix $B$ of the corresponding formation graph $\mathcal{G}$ is 
\[B=\left[\begin{array}{ccccc}
	1 & 0 & 0 & -1 & 1 \\
	-1 & 1 & 0 & 0 & 0 \\
	0 & -1 & 1 & 0 & -1 \\
	0 & 0 & -1 & 1 & 0
\end{array}\right].\]
The manipulators' base are fixed at $x_{10}=[0, 0]^\mathrm{T}$, $x_{20}=[5, 0]^\mathrm{T}$, $x_{30}=[5, 3]^\mathrm{T}$ and $x_{40}=[0, 3]^\mathrm{T}$, respectively. {\bbb The manipulators start from $q_1(0) = [-\pi/2, \pi/3]^\mathrm{T}$, $q_2(0)  = [\pi/6, \pi/3]^\mathrm{T}$, $q_3(0)  = [\pi/2, \pi/3]^\mathrm{T}$ and $q_4(0)  = [-\pi/2, -\pi/3]^\mathrm{T}$ with zero initial joint velocities, respectively.} We use the distributed formation controller \eqref{eq:controller}  with the controller gains of $K_P = 800, K_D = 600$. Correspondingly, we consider three numerical cases:
\begin{enumerate}
\item Case 1: Manipulators 1--3 are fully-actuated and manipulator 4 is the  PA manipulator. The result is shown in Figs. \ref{formation}--\ref{states}.
\item Case 2: Manipulators 1 and 2 are fully-actuated, and manipulators 3 and 4 are the  PA manipulators. The result is shown in Fig. \ref{formation_two}.
\item Case 3: Manipulator 1 is fully-actuated and manipulators 2--4 are the PA manipulators. The result is shown in Fig. \ref{formation_three}.
\end{enumerate}  

For Case 1, consider the Jacobian matrix $\bar J_4$ associated to the PA manipulator 4. {\bb As in Remark \ref{remark}, $\bar J_4$ can be expressed as a function of $q_{4,2}$ only.} We can calculate numerically that $\bar J_{4,1} \bar J_{4,2}=0$ has only three solutions at  $q_{4,2} = -1.3098, 0.2137, 2.2972$ for all $q_{4,2} \in (-\pi, \pi)$. That means Assumption \ref{Ass:3} is satisfied if the manipulator 4 does not go through these three points {\bb and every manipulator $i\in \{1,2,3\}$ does not go through $q_{i,2} = k\pi, k \in \mathbb{Z}$}. 
Fig. \ref{formation} shows the trajectories of the manipulators’ end-effector and Fig. \ref{distance} shows that the distances between the end-effectors converge to expected values. From Figs. \ref{formation} and \ref{distance}, we observe that end-effectors in the network eventually form the expected formation shape. From Fig. \ref{states}, which shows all joint position and velocity signals, we notice that $-1.050<q_{4,2}<-1.005$ and $0.950<q_{i,2}<1.050$ ($i=1,2,3$). 
That means the manipulators do not go through the calculated singular points and Assumption \ref{Ass:3} is satisfied.

Figs. \ref{formation_two} and \ref{formation_three} show that the distributed controller is also effective for Cases 2 and 3, including more underactuated PA manipulators in the network. Notice that the fully-actuated manipulators move in a 2D plane while the  PA manipulators can only move in a line. 

\begin{figure}[H]
	\centering
	\includegraphics[scale=0.4]{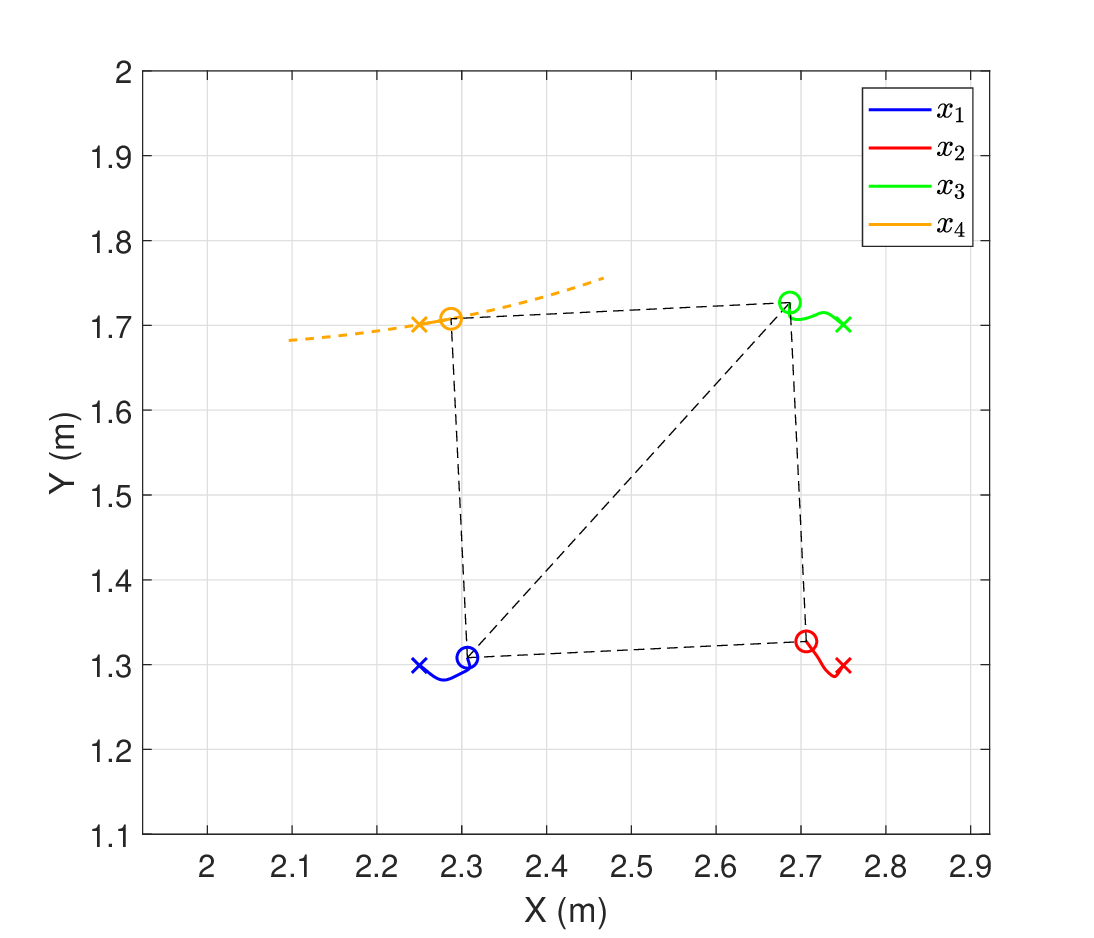}
	\caption{Simulation result of the distributed end-effector formation controller for Case 1. The colored solid lines are the trajectories of each manipulator's end-effector, where $\times$ and $\circ$ denote initial positions and final positions respectively. The colored dashed line is the subset of the workspace of the PA manipulator.}
 \label{formation}
\end{figure}

\begin{figure}[H]
	\centering
	\includegraphics[scale=0.4]{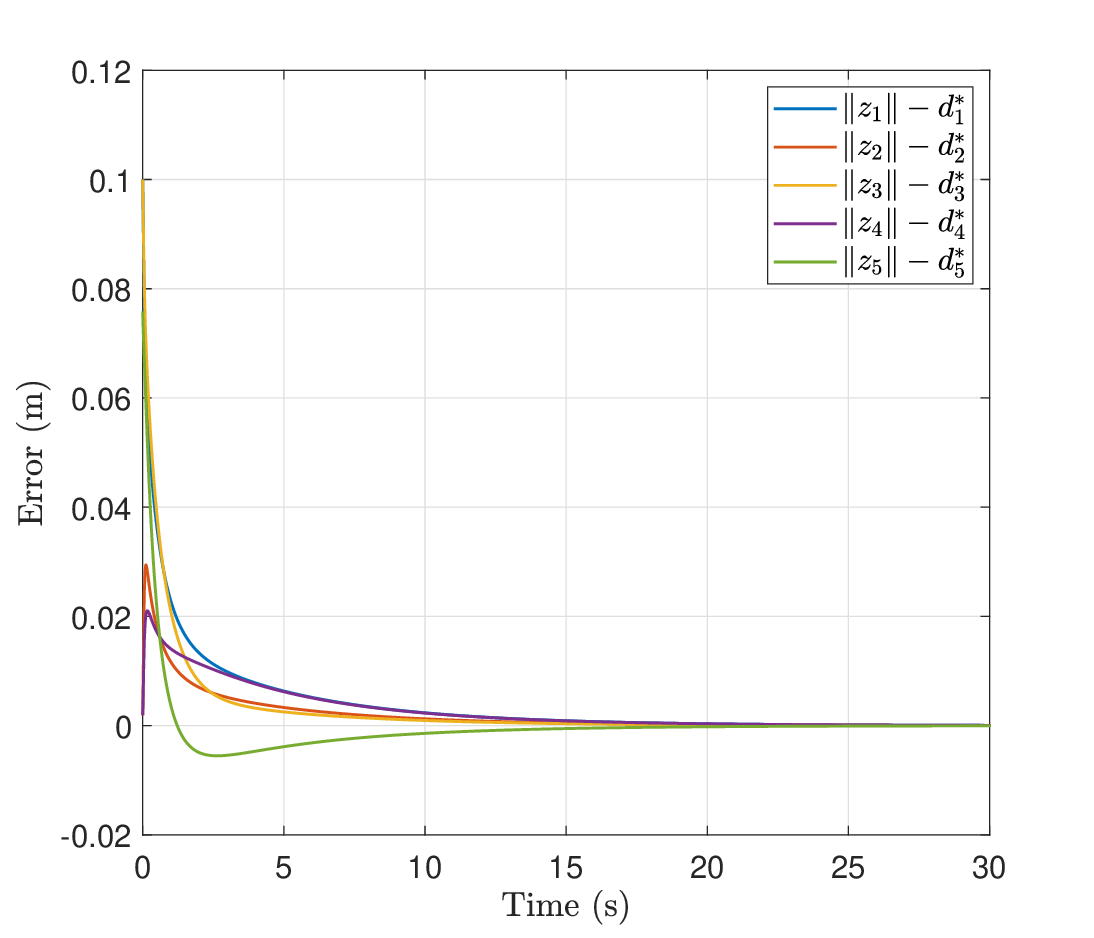}
	\caption{The plot of inter-agent distance error signals for the simulation of the distributed end-effector formation controller in Case 1.}
  \label{distance}
\end{figure}

\begin{figure}[H]
	\centering
	\includegraphics[scale=0.4]{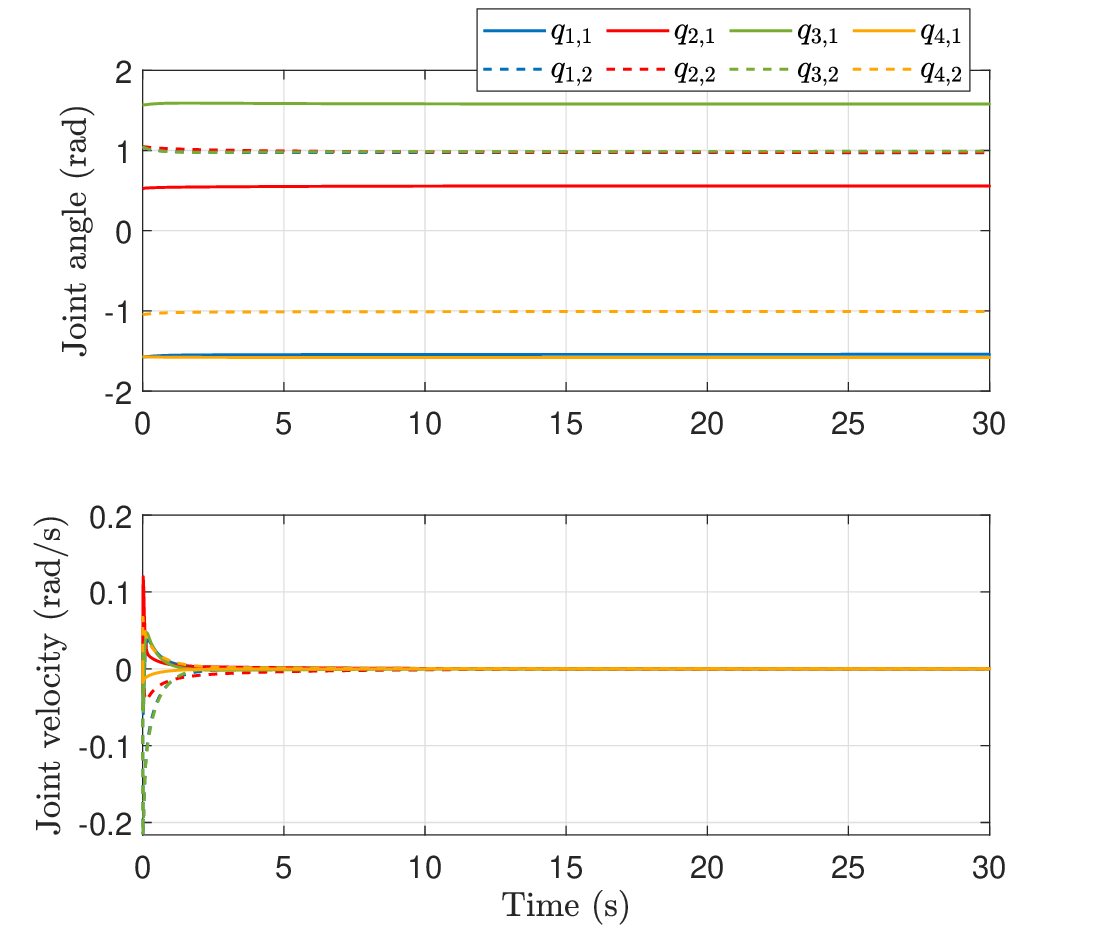}
	\caption{{\bbb The plot of all joint position and velocity signals for the simulation of the distributed end-effector formation controller in Case 1.}}
 \label{states}
\end{figure}

\begin{figure}[H]
	\centering
	\includegraphics[scale=0.4]{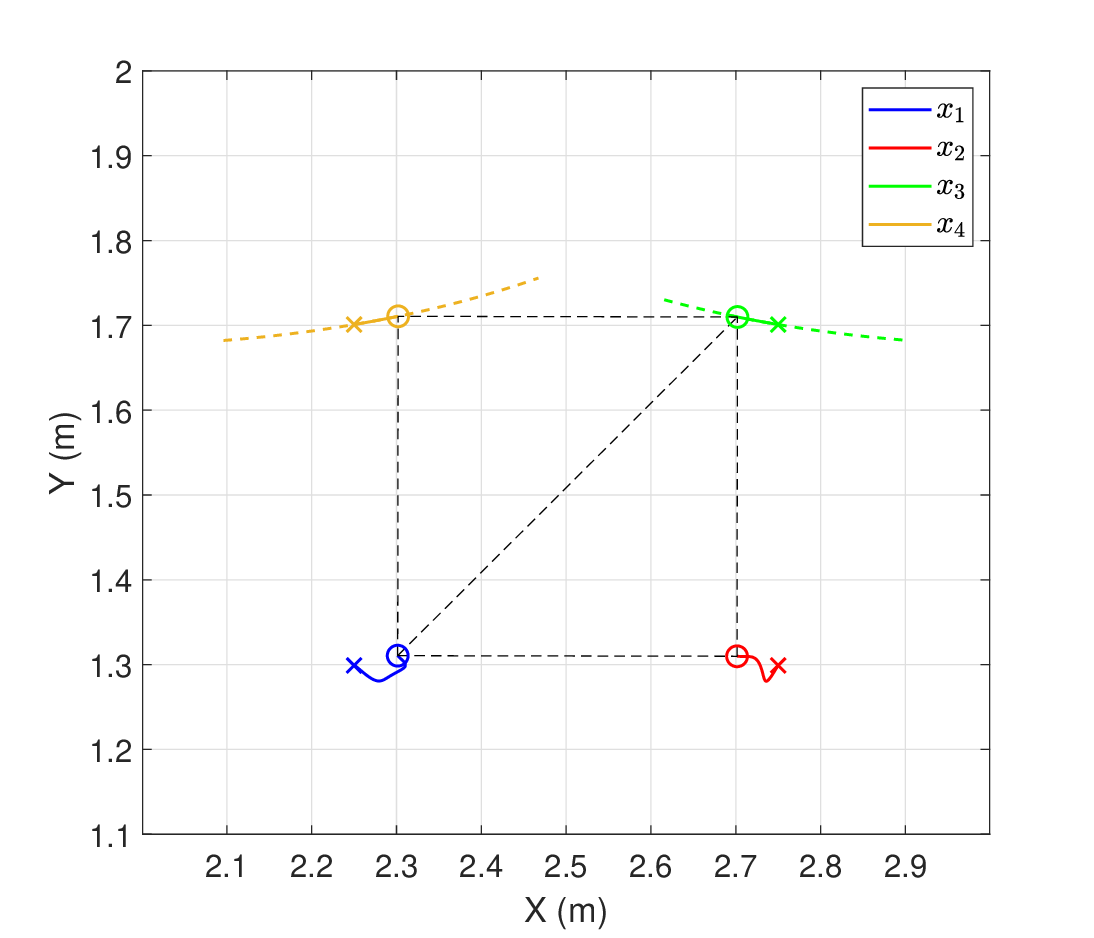}
	\caption{Simulation result of the distributed end-effector formation controller for Case 2. The colored solid lines are the trajectories of each manipulator's end-effector, where $\times$ and $\circ$ denote initial positions and final positions respectively. The colored dashed lines are subsets of the workspace of each PA manipulator.}
  \label{formation_two}
\end{figure}

\begin{figure}[H]
	\centering
	\includegraphics[scale=0.4]{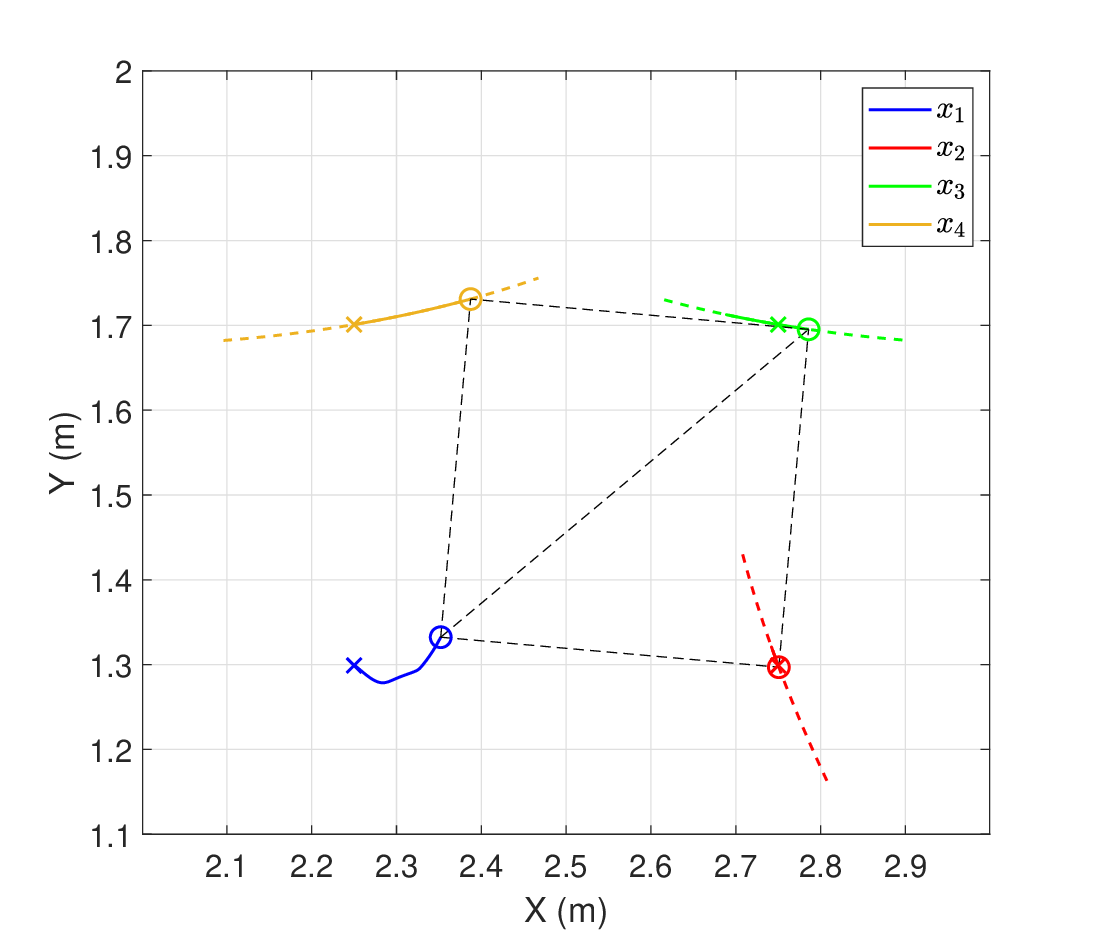}
	\caption{Simulation result of the distributed end-effector formation controller for Case 3. The colored solid lines are the trajectories of each manipulator's end-effector, where $\times$ and $\circ$ denote initial positions and final positions respectively. The colored dashed lines are subsets of the workspace of each PA manipulator.}
   \label{formation_three}
\end{figure}

\section{Conclusion}
This paper studied the end-effector distributed 
formation control for a mixed group of two-link manipulators moving in a horizontal plane, which comprises of fully-actuated manipulators and PA manipulators. Using the integrability property of the PA manipulator, we proposed and analyzed the distributed formation controller for end-effectors. 

\label{conclusion}
 
\bibliographystyle{IEEEtran}
\bibliography{pzy}

\end{document}